\documentclass[12pt,a4paper,oneside,reqno]{amsart}

\usepackage{lmodern}
\usepackage[margin=3cm]{geometry}
\usepackage{setspace}
\usepackage{amssymb}
\usepackage{color}
\usepackage[normalem]{ulem} 
\usepackage{mathtools}
\usepackage{lineno}
\usepackage{hyperref}

\author{Gregory J. Galloway}
\address{Department of Mathematics, University of Miami, Coral Gables, FL, USA.}
\email{galloway@math.miami.edu}

\author{Abraão Mendes}
\address{Instituto de Matemática, Universidade Federal de Alagoas, Maceió, AL, Brazil.}
\email{abraao.mendes@im.ufal.br}

\title[Some rigidity results for compact initial data sets]{Some rigidity results for compact \\ initial data sets}

\newtheorem{thm}{Theorem}[section]
\newtheorem{prop}[thm]{Proposition}
\newtheorem{cor}[thm]{Corollary}
\theoremstyle{remark}
\newtheorem{rmk}[thm]{Remark}

\DeclareMathOperator{\Ric}{Ric}
\DeclareMathOperator{\inn}{in}
\DeclareMathOperator{\out}{out}
\DeclareMathOperator{\id}{id}
\DeclareMathOperator{\Real}{Re}
\DeclareMathOperator{\tr}{tr}
\DeclareMathOperator{\divergence}{div}

\newcommand{\A}{\mathcal{A}}
\newcommand{\B}{\mathcal{B}}
\newcommand{\R}{\mathbb{R}}
\newcommand{\V}{\mathcal{V}}
\newcommand{\p}{\partial}

\renewcommand{\Re}{\Real}
\renewcommand{\sl}[1]{\textsl{#1}}
\renewcommand{\div}{\divergence}

\newcommand{\bbR}{\mathbb{R}}

\begin{document}

\raggedbottom

\numberwithin{equation}{section}

\setstretch{1.1}

\begin{abstract}
In this paper, we prove several rigidity results for compact initial data sets, in both the boundary and no boundary cases. In particular, under natural energy, boundary, and topological conditions, we obtain a global version of the main result in \cite{GalMen}. We also obtain some extensions of results in \cite{EicGalMen}. A number of examples are given in order to illustrate some of the results presented in this paper.
\end{abstract}

\maketitle

\section{Introduction}

The theory of marginally outer trapped surfaces has played an important role in several areas of mathematical general relativity, for example, in proofs of the spacetime positive mass theorem (\sl{e.g.} \cite{EHLS,LeeLesUng}) and in results on the topology of black holes (\sl{e.g.} \cite{GalMots2,GalSch}). In \cite{GalMots2}, a local MOTS rigidity result was obtained, which implies that an outermost MOTS (\sl{e.g.} the surface of a black hole) in an initial data set satisfying the dominant energy condition ($\mu\ge|J|$) is positive Yamabe, \sl{i.e.} admits a metric of positive scalar curvature. This in turn leads to well-known restrictions on the topology of $3$-dimensional outermost MOTS. Such results extend to the spacetime setting well-known results concerning Riemannian manifolds of nonnegative scalar curvature.

In \cite{EicGalMen}, the authors, together with M.~Eichmair, obtained, among other results, a global version of the local MOTS rigidity result in \cite{GalMots2}, which, in particular, does not require a weakly outermost condition; see \cite[Theorem~1.2]{EicGalMen}. This result was motivated in part by J.~Lohkamp's approach to the spacetime positive mass theorem in \cite{Lohkamp2016}. It implies, in dimensions $3\le n\le 7$, Lohkamp's result on the nonexistence of `$\mu-|J|>0$ islands', \cite[Theorem~2]{Lohkamp2016}. Theorem~1.2 in \cite{EicGalMen} has also been applied to obtain a positive mass theorem for asymptotically hyperbolic manifolds with boundary; see \cite{ChruGal}. This theorem will be a useful tool in the present work as well.

In this paper, we present some further initial data rigidity results for compact initial data sets, in both the boundary and no boundary cases. In \cite{GalMen}, the authors considered $3$-dimensional initial data sets containing spherical MOTS. It was shown, roughly speaking, that in a matter-filled spacetime, perhaps with positive cosmological constant, a stable marginally outer trapped $2$-sphere must satisfy a certain area inequality; namely, its area must be bounded above by $4\pi/c$, where $c > 0$ is a lower bound on a natural energy-momentum term. We then established rigidity results for stable, or weakly outermost, marginally outer trapped $2$-spheres when this bound is achieved. In particular, we prove a local splitting result, \cite[Theorem~3.2]{GalMen}, that extends to the spacetime setting a result of H.~Bray, S.~Brendle, and A.~Neves~\cite{BraBreNev} concerning area minimizing $2$-spheres in Riemannian $3$-manifolds with positive scalar curvature. These spacetime results have interesting connections to the Vaidya and Nariai spacetimes \cite{GalMen}.

One of the main aims of the present work is to obtain a global version of \cite[Theorem~3.2]{GalMen}; see Theorem~\ref{thm:main.boundary} in Section~\ref{sec:boundary.cases} for a statement. The proof makes use of certain techniques introduced in \cite{EicGalMen}. In this work, we have also been led to consider certain variations of \cite[Theorem~5.2]{EicGalMen}; see Theorems~\ref{thm:Brane.1}~and~\ref{thm:Brane.2} in Section~\ref{sec:boundary.cases}. Here, it becomes useful to consider the so-called `brane action', as well as the area functional. These results are then used to examine the question of the existence of MOTS in closed (compact without boundary) initial data sets in Section~\ref{sec:closed.cases}. The relationship to known spacetimes is also discussed.

The paper is organized as follows: in Section~\ref{sec:preliminaries}, we review some background material on MOTS; in Section~\ref{sec:boundary.cases}, we state and prove several global rigidity results for compact-with-boundary initial data sets; and, in Section~\ref{sec:closed.cases}, we apply the results obtained in Section~\ref{sec:boundary.cases} to prove some global rigidity statements for closed initial data sets. In Section~\ref{sec:closed.cases}, we also give various examples in order to illustrate the results presented in this paper.

\medskip
\paragraph{\bf{Acknowledgements.}} The work of GJG was partially supported by the Simons Foundation, under Award No. 850541. The work of AM was partially supported by the Conselho Nacional de Desenvolvimento Científico e Tecnológico - CNPq, Brazil (Grant 305710/2020-6), the Coordenação de Aperfeiçoamento de Pessoal de Nível Superior - CAPES, Brazil (CAPES-COFECUB 88887.143161/2017-0), and the Fundação de Amparo à Pesquisa do Estado de Alagoas - FAPEAL, Brazil (Process E:60030.0000002254/2022). The authors would like to thank Ken Baker and Da Rong Cheng for helpful comments. They would also like to thank Christina Sormani for a comment that motivated the results in Section~\ref{sec:closed.cases}.

\section{Preliminaries}\label{sec:preliminaries}

All manifolds in this paper are assumed to be connected and orientable except otherwise stated.

An \sl{initial data set} $(M,g,K)$ consists of a Riemannian manifold $(M,g)$ with boundary $\p M$ (possibly $\p M=\varnothing$) and a symmetric $(0,2)$-tensor $K$ on $M$. 

Let $(M,g,K)$ be an initial data set.

The \sl{local energy density} $\mu$ and the \sl{local current density} $J$ of $(M,g,K)$ are given by
\begin{align*}
\mu=\frac{1}{2}(S-|K|^2+(\tr K)^2)\quad\mbox{and}\quad J=\div(K-(\tr K)g),
\end{align*}
where $S$ is the scalar curvature of $(M,g)$. We say that $(M,g,K)$ satisfies the \sl{dominant energy condition} (DEC for short) if
\begin{align*}
\mu\ge|J|\quad\mbox{on}\quad M.
\end{align*}

Consider a closed embedded hypersurface $\Sigma\subset M$. Since, by assumption, $\Sigma$ and $M$ are orientable, we can choose a unit normal field $\nu$ on $\Sigma$. If $\Sigma$ separates $M$, by convention, we say that $\nu$ points to the outside of $\Sigma$.

The \sl{null second fundamental forms} $\chi^+,\chi^-$ of $\Sigma$ in $(M,g,K)$ with respect to $\nu$ are given by
\begin{align*}
\chi^+=K|_\Sigma+A\quad\mbox{and}\quad\chi^-=K|_\Sigma-A,
\end{align*}
where $A$ is the second fundamental form of $\Sigma$ in $(M,g)$ with respect to $\nu$. More precisely,
\begin{align*}
A(X,Y)=g(\nabla_X\nu,Y)\quad\mbox{for}\quad X,Y\in\mathfrak{X}(\Sigma),
\end{align*}
where $\nabla$ is the Levi-Civita connection of $(M,g)$.

The \sl{null expansion scalars} $\theta^+,\theta^-$ of $\Sigma$ in $(M,g,K)$ with respect to $\nu$ are given by
\begin{align}\label{eq:MOTS.def}
\theta^+=\tr_\Sigma(K)+H\quad\mbox{and}\quad\theta^-=\tr_\Sigma(K)-H,
\end{align}
where $H=\tr A$ is the mean curvature of $\Sigma$ in $(M,g)$ with respect to $\nu$. Observe that $\theta^\pm=\tr\chi^\pm$.

R.~Penrose introduced the now famous notion of a \sl{trapped surface}, when both $\theta^+$ and $\theta^-$ are negative. Restricting to one side, we say that $\Sigma$ is \sl{outer trapped} if $\theta^+<0$, \sl{weakly outer trapped} if $\theta^+\le0$, and \sl{marginally outer trapped} if $\theta^+=0$. In the latter case, we refer to $\Sigma$ as a \sl{marginally outer trapped surface} (MOTS for short). 

Assume now that $\Sigma$ is a MOTS in $(M,g,K)$, with respect to a unit normal $\nu$, that is a boundary in $M$. More precisely, assume that $\nu$ points towards a top-dimensional submanifold $M^+\subset M$ such that $\p M^+=\Sigma\sqcup S$, where $S$ (possibly $S=\varnothing$) is a union of components of $\p M$ (in particular, if $\Sigma$ separates $M$). We think of $M^+$ as the region outside of $\Sigma$. Then we say that $\Sigma$ is \sl{outermost} (resp. \sl{weakly outermost}) if there is no closed embedded hypersurface in $M^+$ with $\theta^+\le0$ (resp. $\theta^+<0$) that is homologous to and different from $\Sigma$. The notions of \sl{locally weakly outermost} and \sl{locally outermost} MOTS can be given in an analogous way.

\begin{rmk}
It is important to mention that initial data sets arise naturally in general relativity. In fact, let $M$ be a spacelike hypersurface in a \sl{spacetime}, \sl{i.e.} a time-oriented Lorentzian manifold, $(\bar N,\bar h)$. Let $g$ be the Riemannian metric on $M$ induced from $\bar h$ and $K$ be the second fundamental form of $M$ in $(\bar N,\bar h)$ with respect to the future-pointing timelike unit normal $u$ on $M$. Then $(M,g,K)$ is an initial data set. As before, let $\Sigma$ be a closed embedded hypersurface in $M$. In this setting, $\chi^+$ and $\chi^-$ are the \sl{null second fundamental forms} of $\Sigma$ in $(\bar N,\bar h)$ with respect to the \sl{null normal fields} 
\begin{align*}
\ell^+=u|_\Sigma+\nu\quad\mbox{and}\quad\ell^-=u|_\Sigma-\nu,
\end{align*}
respectively. Observe that $\theta^\pm=\div_\Sigma\ell^\pm$. Physically, $\theta^+$ (resp. $\theta^-$) measures the divergence of the outward pointing (resp. inward pointing) light rays emanating from $\Sigma$.
\end{rmk}

An initial data set $(M,g,K)$ is said to be \sl{time-symmetric} or \sl{Riemannian} if $K=0$. In this case, a MOTS in $(M,g,K)$ is nothing but a minimal hypersurface in $(M,g)$. Moreover, the energy condition $\mu-|J|\ge c$, for some constant $c$, reduces to the requirement on the scalar curvature $S\ge2c$. Quite generally, marginally outer trapped surfaces share many properties with minimal hypersurfaces, which they generalize; see \sl{e.g.} the survey article \cite{AndEicMet}. 

As in the minimal hypersurfaces case, an important notion for the theory of MOTS is the notion of \sl{stability} introduced, in the context of MOTS, by L.~Andersson, M.~Mars, and W.~Simon~\cite{AndMarSim05,AndMarSim08}, which we now recall.

Let $\Sigma$ be a MOTS in $(M,g,K)$ with respect to $\nu$. Consider a normal variation of $\Sigma$ in $M$, \sl{i.e.} a variation $t\rightarrow\Sigma_t$ of $\Sigma=\Sigma_0$ with variation vector field $\frac{\p}{\p t}|_{t=0}=\phi\,\nu$, $\phi\in C^\infty(\Sigma)$. Let $\theta^\pm(t)$ denote the null expansion scalars of $\Sigma_t$ with respect to $\nu_t$, $\nu=\nu_t|_{t=0}$. Computations as in \cite[p.~2]{AndMarSim05} or \cite[p.~861]{AndMarSim08} give,
\begin{align}\label{eq:first.theta}
\frac{\p\theta^+}{\p t}\Big|_{t=0}=L\phi,
\end{align} 
where
\begin{align*}
L\phi=-\Delta\phi+2\langle X,\nabla\phi\rangle+(Q+\div X-|X|^2)\phi
\end{align*}
and
\begin{align*}
Q=\frac{1}{2}S_\Sigma-(\mu+J(\nu))-\frac{1}{2}|\chi^+|^2.
\end{align*}
Here, $\Delta$ is the negative semi-definite Laplace-Beltrami operator, $\nabla$ the gradient, $\div$ the divergence, and $S_\Sigma$ the scalar curvature of $\Sigma$ with respect to the induced metric $\langle\,\cdot\,,\,\cdot\,\rangle$. Moreover, $X$ is the tangent vector field on $\Sigma$ that is dual to the 1-form $K(\nu,\,\cdot\,)|_\Sigma$.

It is possible to prove (see \cite[Lemma~4.1]{AndMarSim08}) that $L$ has a real eigenvalue $\lambda_1=\lambda_1(L)$, called the \sl{principal eigenvalue} of $L$, such that $\Re\lambda\ge\lambda_1$ for any other complex eigenvalue $\lambda$. Furthermore, the corresponding eigenfunction $\phi_1$, $L\phi_1=\lambda_1\phi_1$, is unique up to a multiplicative constant and can be chosen to be real and everywhere positive.

Then a MOTS $\Sigma$ is said to be \sl{stable} if $\lambda_1(L)\ge0$. This is equivalent to the existence of a positive function $\phi\in C^\infty(\Sigma)$ such that $L\phi\ge0$. It follows directly from \eqref{eq:first.theta} with $\phi=\phi_1$ that every locally weakly outermost (in particular, locally outermost) MOTS is stable.

Observe that in the Riemannian case, $L$ reduces to the classical stability operator, also known as the Jacobi operator, for minimal hypersurfaces. As such, in the literature, $L$ is known as the \sl{MOTS stability operator} or the \sl{stability operator for MOTS}.

The study of rigidity results for minimal surfaces in Riemannian manifolds with a lower scalar curvature bound has been, and continues to be, an active area of research. From the point of view of initial data sets, these are time-symmetric results, as noted above. It has been of interest to extend some of these results to general initial data sets. In the context of general relativity, black hole horizons within initial data sets are often modeled by MOTS, and, in particular, minimal surfaces in the time-symmetric case. These rigidity results often shed light on properties of spacetimes with black holes, as noted in the introduction.

The next proposition and theorem extend to the general non-time-symmetric setting some results of Bray, Brendle, and Neves \cite{BraBreNev}.

\begin{prop}[Infinitesimal rigidity, \cite{GalMen}]\label{prop:infinitesimal.rigidity}
Let $\Sigma$ be a stable MOTS in a $3$-dimensional initial data set $(M,g,K)$ with respect to a unit normal field $\nu$. Suppose there exists a constant $c>0$ such that $\mu+J(\nu)\ge c$ on $\Sigma$. Then the area of $\Sigma$ satisfies,
\begin{align*}
\A(\Sigma)\le\frac{4\pi}{c}.
\end{align*}
Moreover, if $\A(\Sigma)=4\pi/c$, then the following hold:
\begin{enumerate}
\item[(a)] $\Sigma$ is a round $2$-sphere with Gaussian curvature $\kappa_\Sigma=c$,
\item[(b)] the second fundamental form $\chi^+$ of $\Sigma$ with respect to $\nu$ vanishes, and
\item[(c)] $\mu+J(\nu)=c$ on $\Sigma$. 
\end{enumerate}
\end{prop}

The proposition above is used in the proof of the following local splitting theorem. But, before stating the next result, which is also used in the proof of Theorem~\ref{thm:main.boundary}, let us remember the notion of an \sl{area minimizing surface}.

With respect to a fixed Riemannian metric $g$ on a $3$-dimensional manifold $M$, a closed embedded surface $\Sigma\subset M$ is said to be \sl{area minimizing} if $\Sigma$ is of least area in its homology class in $M$, that is, $\A(\Sigma)\le\A(\Sigma')$ for any closed embedded surface $\Sigma'$ that is homologous to $\Sigma$ in $M$. In this case, we also say that $\Sigma$ \sl{minimizes area}. Similarly, $\Sigma$ is said to be \sl{locally area minimizing} if $\A(\Sigma)\le\A(\Sigma')$ for any such $\Sigma'$ in a neighborhood of $\Sigma$ in $M$.

\begin{thm}[Local splitting, \cite{GalMen}]\label{thm:local.splitting}
Let $(M,g,K)$ be a $3$-dimensional initial data set with boundary. Suppose that $(M,g,K)$ satisfies the energy condition $\mu-|J|\ge c$ for some constant $c>0$. Let $\Sigma_0$ be a closed connected component of $\p M$ such that the following conditions hold:
\begin{enumerate}
\item $\Sigma_0$ is a MOTS with respect to the normal that points into $M$ and
\item $\Sigma_0$ is locally weakly outermost and locally area minimizing.
\end{enumerate}
Then $\Sigma_0$ is topologically $S^2$ and its area satisfies,
\begin{align*}
\A(\Sigma_0)\le\frac{4\pi}{c}.
\end{align*}
Furthermore, if $\A(\Sigma_0)=4\pi/c$, then a collar neighborhood $U$ of $\Sigma$ in $M$ is such that:
\begin{enumerate}
\item[(a)] $(U,g)$ is isometric to $([0,\delta)\times\Sigma_0,dt^2+g_0)$ for some $\delta>0$, where $g_0$ - the induced metric on $\Sigma_0$ - has constant Gaussian curvature $\kappa_{\Sigma_0}=c$,
\item[(b)] $K=a\,dt^2$ on $U$, where $a\in C^\infty(U)$ depends only on $t\in[0,\delta)$, and
\item[(c)] $\mu=c$ and $J=0$ on $U$.
\end{enumerate}
\end{thm}

This theorem extends to the general non-time-symmetric setting the local rigidity statements in \cite{BraBreNev}. The local rigidity obtained in \cite{BraBreNev} is then used to obtain a global rigidity result; see \cite[Proposition~11]{BraBreNev}. In Theorem~\ref{thm:main.boundary} in the next section, we obtain a global version of Theorem~\ref{thm:local.splitting}. A key improvement in this global rigidity result is that it does not require the `weakly outermost' assumption, and hence parallels somewhat more closely the global result in \cite{BraBreNev}.

Now, we recall two topological concepts that are important for our purposes; see also \cite{EicGalMen}.

We say that $M$ satisfies the \sl{homotopy condition} with respect to $\Sigma\subset M$ provided there exists a continuous map $\rho:M\to\Sigma$ such that $\rho\circ i:\Sigma\to\Sigma$ is homotopic to~$\id_\Sigma$, where $i:\Sigma\hookrightarrow M$ is the inclusion map (for example, if $\Sigma$ is a retract of $M$).

On the other hand, a closed not necessarily connected manifold $N$ of dimension~$m$ is said to satisfy the \sl{cohomology condition} if there are $m$ classes $\omega_1,\ldots,\omega_m$ in the first cohomology group $H^1(N)$, with integer coefficients, whose cup product 
\begin{align*}
\omega_1\smile\cdots\smile\omega_m\in H^m(N)
\end{align*}
is nontrivial. For example, the $m$-torus $T^m=S^1\times\cdots\times S^1$ satisfies the cohomology condition. More generally, the connected sums $T^m\,\sharp\,Q$ satisfy the cohomology condition for any closed $m$-manifolds $Q$. A version of this condition is considered in \cite[Theorem~5.2]{SY2017}. Here, we are using the form of the condition as it appears in \cite[Theorem~2.28]{Lee}. A manifold $N$ satisfying this cohomology condition has a component that does not carry a metric of positive scalar curvature; see the discussion in \cite{Lee}.

We will make use of the following theorem (mentioned in the introduction) in several situations.

\begin{thm}[{\cite[Theorem~1.2]{EicGalMen}}]\label{thm:EicGalMen}
Let $(M,g,K)$ be an $n$-dimensional, $3\le n\le 7$, compact-with-boundary initial data set. Suppose that $(M,g,K)$ satisfies the dominant energy condition, $\mu\ge|J|$. Suppose also that the boundary can be expressed as a disjoint union $\p M=\Sigma_0\cup S$ of nonempty unions of components such that the following conditions hold:
\begin{enumerate}
\item[(1)] $\theta^+\le0$ on $\Sigma_0$ with respect to the normal that points into $M$,
\item[(2)] $\theta^+\ge0$ on $S$ with respect to the normal that points out of $M$,
\item[(3)] $M$ satisfies the homotopy condition with respect to $\Sigma_0$, and
\item[(4)] $\Sigma_0$ satisfies the cohomology condition.
\end{enumerate}
Then the following hold:
\begin{enumerate}
\item[(a)] $M\cong[0,\ell]\times\Sigma_0$ for some $\ell>0$.
\item[] Let $\Sigma_t\cong\{t\}\times\Sigma_0$ with unit normal $\nu_t$ in direction of the foliation.
\item[(b)] $\chi^+=0$ on $\Sigma_t$ for every $t\in[0,\ell]$.
\item[(c)] $\Sigma_t$ is a flat $(n-1)$-torus with respect to the induced metric for every $t\in[0,\ell]$.
\item[(d)] $\mu+J(\nu_t)=0$ on $\Sigma_t$ for every $t\in[0,\ell]$. In particular, $\mu=|J|$ on $M$.
\end{enumerate}
\end{thm}

The following is the basic existence result for MOTS due to L.~Andersson and J.~Metzger in $3$-dimensions, and M.~Eichmair in dimensions $3\le n\le 7$. It is used in the proof of Theorem~\ref{thm:main.boundary}, and is the source of the dimension restriction appearing in various results discussed herein.

\begin{thm}[Existence of MOTS, \cite{AndMet,Eic09,Eic10}]\label{thm:MOTS.existence}
Let $(M,g,K)$ be an $n$-dimensional, $3\le n\le 7$, compact-with-boundary initial data set. Suppose that the boundary can be expressed as a disjoint union $\p M=\Sigma_{\inn}\cup\Sigma_{\out}$, where $\Sigma_{\inn}$ and $\Sigma_{\out}$ are nonempty unions of components of $\p M$ with $\theta^+\le 0$ on $\Sigma_{\inn}$ with respect to the normal pointing into $M$ and $\theta^+>0$ on $\Sigma_{\out}$ with respect to the normal pointing out of $M$. Then there is an outermost MOTS in $(M,g,K)$ that is homologous to $\Sigma_{\out}$.
\end{thm}

\section{The compact-with-boundary cases}\label{sec:boundary.cases}

In this section, we obtain several global initial data results. These results, in turn, will be applied in the next section to the case that the initial data manifold is closed. The first is a global version of Theorem~\ref{thm:local.splitting}; see the comments above, after the statement of Theorem~\ref{thm:local.splitting}.

\begin{thm}\label{thm:main.boundary}
Let $(M,g,K)$ be a $3$-dimensional compact-with-boundary initial data set. Suppose that $(M,g,K)$ satisfies the energy condition $\mu-|J|\ge c$ for some constant $c>0$. Suppose also that the boundary can be expressed as a disjoint union $\p M=\Sigma_0\cup S$ of nonempty unions of components such that the following conditions hold:
\begin{enumerate}
\item[(1)] $\theta^+\le0$ on $\Sigma_0$ with respect to the normal that points into $M$,
\item[(2)] $\theta^+\ge0$ on $S$ with respect to the normal that points out of $M$,
\item[(3)] $M$ satisfies the homotopy condition with respect to $\Sigma_0$,
\item[(4)] the relative homology group $H_2(M,\Sigma_0)$ vanishes, and
\item[(5)] $\Sigma_0$ minimizes area.
\end{enumerate}
Then $\Sigma_0$ is topologically $S^2$ and its area satisfies,
\begin{align*}
\A(\Sigma_0)\le\frac{4\pi}{c}.
\end{align*}
Moreover, if $\A(\Sigma_0)=4\pi/c$, then the following hold:
\begin{enumerate}
\item[(a)] $(M,g)$ is isometric to $([0,\ell]\times\Sigma_0,dt^2+g_0)$ for some $\ell>0$, where $g_0$ - the induced metric on $\Sigma_0$ - has constant Gaussian curvature $\kappa_{\Sigma_0}=c$,
\item[(b)] $K=a\,dt^2$ on $M$, where $a\in C^\infty(M)$ depends only on $t\in[0,\ell]$, and
\item[(c)] $\mu=c$ and $J=0$ on $M$.
\end{enumerate}
\end{thm}


\begin{proof}
First, observe that $\Sigma_0$ is connected, since $M$ is connected and satisfies the homotopy condition with respect to $\Sigma_0$. 

If $\Sigma_0$ is not homeomorphic to $S^2$, then $\Sigma_0$ is homeomorphic to $T^2\,\sharp\,Q$ for some closed orientable surface $Q$. In particular, $\Sigma_0$ satisfies the cohomology condition and so Theorem~\ref{thm:EicGalMen} applies to $(M,g,K)$. Therefore, $0=\mu-|J|\ge c$ on $M$, which is a contradiction. Then $\Sigma_0$ is topologically $S^2$.

{\bf Claim:} $\Sigma_0$ is a weakly outermost MOTS in $(M,g,K)$ of area $\A(\Sigma_0)=4\pi/c$ unless $\A(\Sigma_0)<4\pi/c$. 

Assume that $\A(\Sigma_0)\ge4\pi/c$.

If $\theta_K^+\le0$ is not identically zero on $\Sigma_0$, it follows from \cite[Lemma~5.2]{AndMet} that there is a surface $\Sigma\subset M$ - obtained by a small perturbation of $\Sigma_0$ into $M$ - such that $\theta_K^+<0$ on $\Sigma$ with respect to the normal pointing away from $\Sigma_0$. Let $W$ be the connected compact region bounded by $\Sigma$ and $S$ in $M$. Observe that $\theta_{-K}^+\le0$ on $S$ with respect to the normal that points into $W$ and $\theta_{-K}^+>0$ on $\Sigma$ with respect to the normal that points out of $W$. Applying the MOTS existence theorem (Theorem~\ref{thm:MOTS.existence}), we obtain an outermost MOTS $\Sigma'$ in $(W,g,-K)$ that is homologous to and disjoint from $\Sigma$. Clearly, $\Sigma'$ is homologous to $\Sigma_0$ in $M$.

Without loss of generality, we may assume that each connected component of $\Sigma'$ is homologically nontrivial in $M$. Also, because $H_2(M,\Sigma_0)=0$, $\Sigma'$ is connected.

Since we are assuming that $\Sigma_0$ minimizes area in its homology class, we have
\begin{align*}
\frac{4\pi}{c}\le\A(\Sigma_0)\le\A(\Sigma').
\end{align*}

On the other hand, because $\Sigma'$ is an outermost MOTS in $(W,g,-K)$, in particular stable, the infinitesimal rigidity (Proposition~\ref{prop:infinitesimal.rigidity}) gives that $\A(\Sigma')=4\pi/c$. Therefore, $\Sigma'$ is an area minimizing outermost MOTS in $(W,g,-K)$ of area $\A(\Sigma')=4\pi/c$ and then the local splitting theorem (Theorem~\ref{thm:local.splitting}) applies so that an outer neighborhood of $\Sigma'$ in $W$ is foliated by MOTS, which is a contradiction.

This proves that $\Sigma_0$ is a MOTS in $(M,g,K)$.

Now, we claim that $\Sigma_0$ is weakly outermost in $(M,g,K)$. If not, there is a surface $\Sigma$ that is homologous to $\Sigma_0$ in $M$ and such that $\theta_K^+<0$ on it. Perturbing $\Sigma$ a bit, we may assume that $\Sigma\cap\Sigma_0=\varnothing$. Also, by the strong maximum principle as in \sl{e.g.} \cite[Proposition~2.4]{AndMet} or \cite[Proposition~3.1]{AshGal}, $\Sigma\cap S=\varnothing$.

As before, without loss of generality, we may assume that each connected component of $\Sigma$ is homologically nontrivial in $M$ and, in particular, $\Sigma$ is connected. Let $W$ be the region in $M$ bounded by $\Sigma$ and $S$. Arguing with $(W,g,-K)$ as above, we have a contradiction. Thus $\Sigma_0$ is weakly outermost.

We have then proved that, if $\A(\Sigma_0)\ge4\pi/c$, then $\Sigma_0$ is a weakly outermost MOTS in $(M,g,K)$. In this case, by the infinitesimal rigidity, $\A(\Sigma_0)=4\pi/c$.

This finishes the proof of the Claim.

We have then obtained that $\Sigma_0$ is homeomorphic to $S^2$ and its area satisfies $\A(\Sigma_0)\le4\pi/c$. Furthermore, if $\A(\Sigma_0)=4\pi/c$, then $\Sigma_0$ is an area minimizing weakly outermost MOTS in $(M,g,K)$. In this case, by the local splitting theorem, there is a collar neighborhood $U\cong[0,\delta)\times\Sigma_0$ of $\Sigma_0$ in $M$ such that conclusions (a), (b), and (c) of the theorem hold on $U$. Clearly, $\Sigma_t\cong\{t\}\times\Sigma_0$ converges to a closed embedded MOTS $\Sigma_\delta$ of area $4\pi/c$ as $t\nearrow\delta$. If $\Sigma_\delta\cap S\neq\varnothing$, by the strong maximum principle, $\Sigma_\delta=S$. If $\Sigma_\delta\cap S=\varnothing$, we can replace $\Sigma_0$ by $\Sigma_\delta$ and $M$ by the complement of $U$ and run the process again. The result then follows by a continuity argument.
\end{proof}

The next two theorems make use of the notion of \sl{$(n-1)$-convexity} of a symmetric $(0,2)$-tensor. Imposing such convexity leads to stronger rigidity. 

We say that a symmetric $(0,2)$-tensor $P$ on $(M,g)$ is \sl{$(n-1)$-convex} if, at every point $p\in M$, the sum of the smallest $(n-1)$ eigenvalues of $P$ with respect to $g$ is nonnegative (in particular, if $P$ is positive semi-definite). This is equivalent to the trace of $P$ with respect to any $(n-1)$-dimensional linear subspace of $T_pM$ being nonnegative, for every $p\in M$. In particular, if $P$ is $(n-1)$-convex, then $\tr_\Sigma P\ge0$ for every hypersurface $\Sigma\subset M$. This convexity condition has been used by the second-named author in \cite{Men} and by the authors, together with M.~Eichmair, in \cite{EicGalMen} in related contexts.

Let $(M,g,K)$ be as in Theorem~\ref{thm:EicGalMen}, and let $\Sigma$ be a closed embedded hypersurface homologous to $\Sigma_0$. 
The next theorem makes use of the functional,
\begin{align*}
\B_\epsilon(\Sigma)=\A(\Sigma)-(n-1)\,\epsilon\,\V(\Sigma),\quad\epsilon=0,1,
\end{align*}
where $\A(\Sigma)$ is the area of $\Sigma$ and $\V(\Sigma)$ is the volume of the region bounded by $\Sigma$ and $\Sigma_0$. In the case $\epsilon=0$, we are just talking about the area functional. In the case $\epsilon=1$, we are talking about the functional associated with hypersurfaces of constant mean curvature $n-1$, sometimes referred to as the \sl{brane action} and denoted by $\B$.

The following theorem extends in a couple of directions Theorem~5.2 in \cite{EicGalMen}.

\begin{thm}\label{thm:Brane.1}
Let $(M,g,K)$ be as in Theorem~\ref{thm:EicGalMen}. Assume that 
\begin{enumerate}
\item[(i)] $K+\epsilon\,g$ is $(n-1)$-convex, where $\epsilon=0$ or $\epsilon=1$, and
\item[(ii)] $\Sigma_0$ and $S$ are such that $\B_\epsilon(\Sigma_0)\le\B_\epsilon(S)$.
\end{enumerate}
Then the following hold:
\begin{enumerate}
\item[(a)] $\Sigma_0$ is a flat $(n-1)$-torus with respect to the induced metric $g_0$,
\item[(b)] $(M,g)$ is isometric to $([0,\ell]\times\Sigma_0,dt^2+e^{2\,\epsilon\,t}g_0)$ for some $\ell>0$,
\item[(c)] $K=(1-\epsilon)a-\epsilon\,g$ on $M$, where $a\in C^\infty(M)$ depends only on $t\in[0,\ell]$, and
\item[(d)] $\mu=0$ and $J=0$ on $M$.
\end{enumerate}
\end{thm}

The convexity assumption holds if, in particular, $K$ satisfies, $K\ge-\epsilon\,g$. In the case $\epsilon = 0$, this would apply to cosmological models that are expanding to the future (in all directions).
\begin{proof}
By Theorem~\ref{thm:EicGalMen},
\begin{enumerate}
\item[-] $M\cong[0,\ell]\times\Sigma_0$ for some $\ell>0$, and
\item[-] each leaf $\Sigma_t\cong\{t\}\times\Sigma_0$ is a MOTS with respect to the unit normal $\nu_t$ in direction of the foliation. 
\end{enumerate}
On the other hand, since $K+\epsilon\,g$ is $(n-1)$-convex, we have
\begin{align}\label{eq:aux.1}
H(t)-(n-1)\,\epsilon\le H(t)+\tr_{\Sigma_t}K=0,
\end{align}
where $H(t)=\div_{\Sigma_t}\nu_t$ is the mean curvature of $\Sigma_t$. 

Now, the first variation of $\B_\epsilon$ gives that
\begin{align}\label{eq:aux.2}
\frac{d}{dt}\B_\epsilon(\Sigma_t)=\int_{\Sigma_t}\phi\,(H(t)-(n-1)\,\epsilon)\,d\Sigma_t\le0,
\end{align}
where $\phi=\langle\nu_t,\p_t\rangle$ is the lapse function of the foliation. Therefore, 
$\B_\epsilon(t)=\B_\epsilon(\Sigma_t)$ is a nonincreasing function on $[0,\ell]$ satisfying $\B_\epsilon(0)\le \B_\epsilon(\ell)$. Thus $\B_\epsilon(t)=\B_\epsilon(\Sigma_t)$ is constant. Inequalities \eqref{eq:aux.1} and \eqref{eq:aux.2} give that $H(t)=(n-1)\,\epsilon=-\tr_{\Sigma_t}K$ for all $t\in[0,\ell]$. In particular, $\theta^-=-2\,(n-1)\,\epsilon$ on $\Sigma_\ell=S$. 

The result then follows directly from \cite[Theorem~1.3]{EicGalMen} (observe that our sign convention in the definition of $\theta^-$ in this work is the opposite of that one in \cite{EicGalMen}).
\end{proof}

In the next theorem we consider $\B=\B_1$ under a modified convexity condition.

\begin{thm}\label{thm:Brane.2}
Let $(M,g,K)$ be as in Theorem~\ref{thm:EicGalMen}. Assume that $-(K+g)$ is $(n-1)$-convex. Then $\B(\Sigma_0)\le\B(S)$. Moreover, if equality holds, we have the following:
\begin{enumerate}
\item [(a)] $(M,g)$ is isometric to $([0,\ell]\times\Sigma_0,dt^2+g_t)$ for some $\ell>0$, where $g_t$ is the induced metric on $\Sigma_t\cong\{t\}\times\Sigma_0$.
\item[(b)] Each $\Sigma_t$ is a flat $(n-1)$-torus with respect to $g_t$ and has constant mean curvature $H(t)=n-1$.
\item[(c)] The scalar curvature of $(M,g)$ satisfies $S\le-n(n-1)$. If equality holds, $(M,g)$ is isometric to $([0,\ell]\times\Sigma_0,dt^2+e^{2\,t}g_0)$.
\item[(d)] For each $t\in[0,\ell]$, $\mu+J(\nu_t)=0$ on $\Sigma_t$. In particular, $\mu=|J|$ on $M$.
\item[(e)] $\tr K\le-n$ on $M$. If equality holds, $K=-g$, $S=-n(n-1)$, $\mu=0$, and $J=0$ on $M$. 
\end{enumerate}
\end{thm}

The convexity assumption holds if, in particular, $K$ satisfies, $K\le-g$. If one views $K$ as being defined with respect to the past directed unit normal, this would apply to cosmological models that are strongly contracting to the past, \sl{e.g.} that begin with a `big bang'.

\begin{proof}
By Theorem~\ref{thm:EicGalMen},
\begin{itemize}
\item[-] $M\cong[0,\ell]\times\Sigma_0$ with
\begin{align}\label{eq:metric}
g=\phi^2dt^2+g_t,
\end{align}
where $g_t$ is the induced metric on $\Sigma_t\cong\{t\}\times\Sigma_0$.
\item[-] Each $(\Sigma_t,g_t)$ is a flat $(n-1)$-torus.
\item[-] Every leaf $\Sigma_t$ is a MOTS in $(M,g,K)$. In fact,
\begin{align*}
0=\chi^+(t)=A(t)+K|_{\Sigma_t},
\end{align*}
where $A(t)$ is the second fundamental form of $\Sigma_t$ computed with respect to the unit normal $\nu_t$ in direction of the foliation.
\item[-] For each $t\in[0,\ell]$, $\mu+J(\nu_t)=0$ on $\Sigma_t$. In particular, $\mu=|J|$ on $M$.
\end{itemize}

Now, since $-(K+g)$ is $(n-1)$-convex, we have
\begin{align}\label{eq:aux.3}
H(t)-(n-1)\ge H(t)+\tr_{\Sigma_t}K=0,
\end{align}
where $H(t)=\tr A(t)$ is the mean curvature of $\Sigma_t$. Then the first variation of $\B$ gives that
\begin{align}\label{eq:aux.4}
\frac{d}{dt}\B(\Sigma_t)=\int_{\Sigma_t}\phi\,(H(t)-(n-1))\,d\Sigma_t\ge0.
\end{align}
Therefore, $\B(t)=\B(\Sigma_t)$ is a nondecreasing function defined on $[0,\ell]$. In particular, $\B(0)\le\B(\ell)$, that is, $\B(\Sigma_0)\le\B(S)$.

If $\B(\Sigma_0)=\B(S)$, then $\B(t)=\B(\Sigma_t)$ is constant. Therefore, inequalities \eqref{eq:aux.3} and~\eqref{eq:aux.4} imply that $H(t)=n-1=-\tr_{\Sigma_t}K$ for all $t\in[0,\ell]$.

Now, fix $t\in[0,\ell]$, $p\in\Sigma_t$, and let $\{e_1,\ldots,e_{n-1}\}$ be an orthonormal basis for $T_p\Sigma_t$. Define
\begin{align*}
\eta(s)=\cos s\cdot e_{n-1}+\sin s\cdot\nu_t,\quad s\in\R,
\end{align*}
and let $\pi(s)$ be the $(n-1)$-dimensional linear subspace of $T_pM$ generated by
\begin{align*}
\{e_1,\ldots,e_{n-2},\eta(s)\}.
\end{align*}
Since $-(K+g)$ is $(n-1)$-convex and $\tr_{\Sigma_t}K=-(n-1)$, we have
\begin{align*}
f(s)\coloneqq\tr_{\pi(s)}K\le-(n-1)\quad\mbox{and}\quad f(0)=\tr_{\Sigma_t}K=-(n-1).
\end{align*}
Therefore, $s=0$ is a critical point of $f(s)$. Observing that 
\begin{align*}
f(s)=\sum_{i=1}^{n-2}K(e_i,e_i)+K(\eta(s),\eta(s)),
\end{align*}
we obtain,
\begin{align*}
0=f'(0)=2K(\eta'(0),\eta(0))=2K(\nu_t,e_{n-1}).
\end{align*}
Analogously, $K(\nu_t,e_i)=0$ for $i=1,\ldots,n-2$. This gives that $X^\flat=K(\nu_t,\,\cdot\,)|_{\Sigma_t}=0$ for all $t\in[0,\ell]$.

On the other hand, the first variation of $\theta^+(t)=0$ reads as
\begin{align*}
\frac{\p\theta^+}{\p t}=-\Delta\phi+2\langle X,\nabla\phi\rangle+(Q+\div X-|X|^2)\phi=-\Delta\phi+Q\phi,
\end{align*}
where 
\begin{align*}
Q=\frac{1}{2}S_{\Sigma_t}-(\mu+J(\nu_t))-\frac{1}{2}|\chi^+(t)|^2=0.
\end{align*}
Thus $\Delta\phi=0$ on $\Sigma_t$ and then $\phi=\phi(t)$ is constant on $\Sigma_t$ for each $t\in[0,\ell]$. Hence, by a simple change of variable in \eqref{eq:metric}, we have
\begin{align}\label{eq:metric.2}
g=dt^2+g_t.
\end{align}
In particular, the $t$-lines are geodesics. Hence, along each leaf $\Sigma=\Sigma_t$, $H=H(t)$ satisfies the scalar Riccati equation,
\begin{align*}
\frac{\p H}{\p t}=-\Ric(\p_t,\p_t)-|A|^2,
\end{align*}
which, since $H(t)=n-1$, implies,
\begin{align*}
\Ric(\p_t,\p_t)+|A|^2=0.
\end{align*}
By the Gauss equation, we have the standard rewriting of the left-hand side in the above equation,
\begin{align*}
\Ric(\p_t,\p_t)+|A|^2=\frac12(S-S_{\Sigma}+|A|^2+H^2).
\end{align*}
Hence, since $S_{\Sigma}=0$, we have,
\begin{align}\label{eq:aux.6}
S=-|A|^2-H^2\le-\frac{H^2}{n-1}-H^2=-n(n-1),
\end{align}
which establishes the inequality part in (c). If equality holds, then $|A(t)|^2=n-1$, which, together with $H(t)=n-1$, implies that each $\Sigma_t$ is umbilic; in fact, $A(t)=g_t$. Using this in \eqref{eq:metric.2} easily implies the isometry part in (c).

Since $-(K+g)$ is $(n-1)$-convex, it is not difficult to see that $\tr K\le -n$. In fact, if $\{e_1,\ldots,e_n\}$ is an orthonormal basis for $T_pM$, $p\in M$, then
\begin{align}\label{eq:aux.5}
(n-1)\tr K=\sum_{i=1}^n\sum_{j\neq i}K(e_j,e_j)\le-\sum_{i=1}^n(n-1)=-n(n-1),
\end{align}
that is, $\tr K\le-n$. If $\tr K=-n$, it follows from \eqref{eq:aux.5} that
\begin{align*}
\sum_{j\neq i}K(e_j,e_j)=-(n-1)\quad\mbox{for each}\quad i=1,\ldots,n.
\end{align*}
Therefore,
\begin{align*}
-n=\tr K=K(e_i,e_i)+\sum_{j\neq i}K(e_j,e_j)=K(e_i,e_i)-(n-1),
\end{align*}
that is, $K(e_i,e_i)=-1$ for each $i=1,\ldots,n$. Since $\{e_1,\ldots,e_n\}$ is arbitrary, we have $K=-g$. Thus, using that $A(t)=-K|_{\Sigma_t}=g|_{\Sigma_t}$ in \eqref{eq:aux.6}, we obtain
\begin{align*}
S=-|A(t)|^2-|H(t)|^2=-n(n-1).
\end{align*}
Finally, 
\begin{align*}
\mu=\frac{1}{2}(S-|K|^2+(\tr K)^2)=0\quad\mbox{and}\quad J=\div(K-(\tr K)g)=0.
\end{align*}
\end{proof}

\section{Applications: closed cases}\label{sec:closed.cases}

In this section, we wish to apply the results of the previous section to initial data manifolds that are closed (compact without boundary). These results naturally relate to cosmological (\sl{i.e.} spatially closed) spacetimes. We'll illustrate the results with various examples.

\subsection{The spherical case}\label{sphericalcase}

In this section, we want to apply Theorem~\ref{thm:main.boundary} to the case that $M$ is closed.

Let $M$ be an $n$-dimensional closed manifold. Suppose the $(n-1)$-th homology group $H_{n-1}(M)$ is nontrivial. Any nontrivial element of $H_{n-1}(M)$ gives rise to a smooth closed embedded non-separating orientable hypersurface $\Sigma\subset M$. In particular, $\Sigma$ is \sl{two-sided} in $M$, \sl{i.e.} there is an embedding $F:[-1,1]\times\Sigma\to M$ such that $F(0,p)=p$ for each $p\in\Sigma$. Let $U$ denote the open set $F((-1,1)\times\Sigma)\subset M$. We say that $M$ is \sl{retractable with respect to $\Sigma$} if $M\setminus U$ retracts onto some component of $\p U$. If we consider a Riemannian metric $g$ on $M$, given a unit normal field $\nu$ on $\Sigma$ with respect to $g$, we say that $M$ is \sl{retractable with respect to $\Sigma$ towards $\nu$} if $M\setminus U$ retracts onto the component of $\p U$ towards which $\nu$ points. 

An obvious situation where this occurs is when $M$ is of the form $M=S^1\times Q$, with $Q$ closed. Then $M$ is retractable with respect to $\Sigma=\{x\}\times Q$, $x\in S^1$. Another situation of interest is when $M$ is of the form $M=T^n\,\sharp\,Q$. View $T^n$ as an $n$-dimensional cube with opposite boundary faces identified. To obtain $M$, we may assume the connected sum takes place in a bounded open set $U$ inside the cube. Let $\Sigma$ be an $(n-1)$-torus parallel to one of the faces away from the set $U$. Then $M$ is retractable with respect to $\Sigma$. More generally, if $M$ is retractable with respect to $\Sigma$, then so is $M\,\sharp\,Q$, with $Q$ closed, provided the connect sum takes place away from $\Sigma$.

\begin{thm}\label{thm:cor.boundary}
Let $(M,g,K)$ be a $3$-dimensional closed initial data set satisfying the energy condition $\mu-|J|\ge c$ for some constant $c>0$. Suppose that $(M,g,K)$ admits a MOTS $\Sigma$, with respect to a unit normal field $\nu$, such that the following conditions hold:
\begin{enumerate}
\item[(I)] $M$ is retractable with respect to $\Sigma$ towards $\nu$,
\item[(II)] the homology group $H_2(M)$ is generated by the class of $\Sigma$, and
\item[(III)] $\Sigma$ minimizes area.
\end{enumerate}
Then $\Sigma$ is topologically $S^2$ and its area satisfies,
\begin{align}\label{ineqA}
\A(\Sigma)\le\frac{4\pi}{c}.
\end{align}
Moreover, if $\A(\Sigma)=4\pi/c$, then the following hold:
\begin{enumerate}
\item[(a')] $(M,g)$ is isometric to $[0,\ell]\times\Sigma/{\sim}$ endowed with the induced metric from the product $([0,\ell]\times\Sigma,dt^2+h)$, where `\,$\sim$' means that $\{0\}\times\Sigma$ and $\{\ell\}\times\Sigma$ are suitably identified and $h$ - the induced metric on $\Sigma$ - has constant Gaussian curvature $\kappa_\Sigma=c$,
\item[(b')] $K=a\,dt^2$ on $M$, where $a\in C^\infty(M)$ depends only on $t$, and
\item[(c')] $\mu=c$ and $J=0$ on $M$.
\end{enumerate}
\end{thm}

\begin{proof}
First, observe that, by making a `cut' along $\Sigma$, we obtain a $3$-dimensional compact manifold $M'$ with two boundary components, say $\Sigma_0$ and $S$. Also, the initial data $(g,K)$ on $M$ gives rise to data $(g',K')$ on $M'$ in the natural way. The boundary components $\Sigma_0$ and $S$ are both isometric to $\Sigma$ with respect to the corresponding induced metrics.

Now, consider the initial data set $(M',g',K')$. Observe that the boundary components $\Sigma_0$ and $S$ of $M'$ can be chosen in such a way that conditions (1)-(5) of Theorem~\ref{thm:main.boundary} are satisfied. In fact,
\begin{enumerate}
\item $\Sigma_0$ is a MOTS with respect to the normal that points into $M'$,
\item $S$ is a MOTS with respect to the normal that points out of $M'$,
\item $M'$ satisfies the homotopy condition with respect to $\Sigma_0$, since $M$ is retractable with respect to $\Sigma$ towards $\nu$,
\item the relative homology group $H_2(M',\Sigma_0)$ vanishes, since $H_2(M)$ is generated by the class of $\Sigma$, and
\item $\Sigma_0$ minimizes area in $(M',g')$ as $\Sigma$ minimizes area in $(M,g)$. 
\end{enumerate}
Conditions (1) and (2) above follow from the fact of $\Sigma$ being a MOTS in $(M,g,K)$ with respect to $\nu$ and the choice of $\Sigma_0$ and $S$. Therefore, by Theorem~\ref{thm:main.boundary}, $\Sigma_0$ is topologically $S^2$ and its area satisfies $\A(\Sigma_0)\le4\pi/c$. The same conclusions hold for~$\Sigma$. Moreover, if $\A(\Sigma)=4\pi/c$, that is, $\A(\Sigma_0)=4\pi/c$, then conclusions (a)-(c) of Theorem~\ref{thm:main.boundary} hold for $(M',g',K')$ and thus $(M,g,K)$ satisfies (a')-(c').
\end{proof}

\begin{rmk} 
Initial data sets satisfying the assumptions of Theorem~\ref{thm:cor.boundary} arise naturally in the Nariai spacetime. The Nariai spacetime is a solution to the vacuum Einstein equations with positive cosmological constant, $\Lambda>0$. It is a metric product of $2$-dimensional de Sitter space $dS_2$ and $S^2$,
\begin{align*}
\bar N=(\bbR\times S^1)\times S^2,\quad\bar h=-dt^2+a^2\cosh^2(t/a)\,d\chi^2+a^2d\Omega^2,
\end{align*}
where $a=\frac{1}{\sqrt{\Lambda}}$. As described in \cite{Bousso, BousHawk}, the Nariai spacetime is an interesting limit of Schwarzschild-de Sitter space, as the size of the black hole increases and its area approaches the upper bound in \eqref{ineqA}, with $c=\Lambda$. 

Under the transformation, $\cosh(t/a)=\sec\tau$, the metric $\bar h$ becomes,
\begin{align*}
\bar h=\frac{a^2}{\cos^2(\tau)}\left(-d\tau^2+d\chi^2\right)+a^2d\Omega^2, 
\end{align*}
where $\tau$ is in the range, $-\frac{\pi}{2}<\tau<\frac{\pi}{2}$. With this change of time coordinate, we see that $dS_2$ is locally conformal to the Minkowski plane. A Penrose type diagram for $(\bar N,\bar h)$ is depicted in Figure~\ref{dsfig1}. Each point in the diagram represents a round $2$-sphere of radius $a$. In the diagram, $M=\Gamma\times S^2$, where $\Gamma$ is a smooth spacelike graph over the circle: $\tau=0$, $0\le\chi\le2\pi$ in $dS_2$. Taking $\Sigma$ to be the $2$-sphere intersection of $M$ with the totally geodesic null hypersurface $H$, one easily verifies that $(M,g,K)$, where $g$ is the induced metric and $K$ is the second fundamental form of $M$, respectively, satisfies the assumptions of Theorem~\ref{thm:cor.boundary}, with equality in \eqref{ineqA}. We note that there are initial data sets in (spatially closed) Schwarzschild-de Sitter that satisfy all the assumptions of Theorem~\ref{thm:cor.boundary}, except for equality in \eqref{ineqA}.
\end{rmk}

\begin{figure}[ht]
\begin{center}
\mbox{
\includegraphics[width=3.2in]{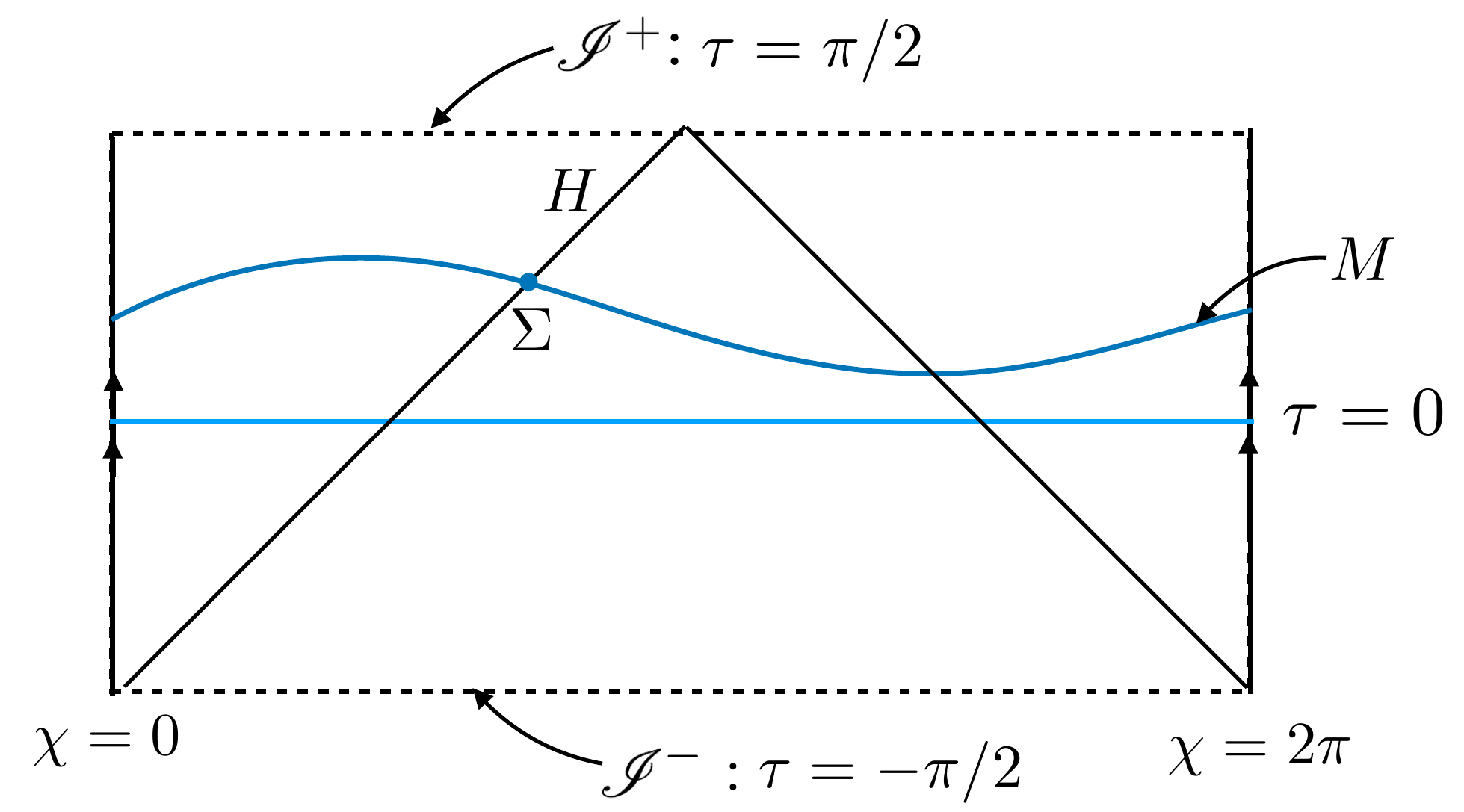}
}
\end{center}
\caption{Nariai spacetime.}
\label{dsfig1}
\end{figure}

\subsection{The non-spherical case}\label{sec:nonsphere}

We now consider applications of Theorems \ref{thm:EicGalMen}, \ref{thm:Brane.1} and \ref{thm:Brane.2} to closed initial data sets satisfying the DEC. As a consequence, we obtain results concerning the existence and rigidity of MOTS with nontrivial (\sl{e.g.} toroidal) topology in the cosmological setting. 

We first consider an example. Let $(\bar N,\bar h)$ be the FLRW spacetime,
\begin{align*}
\bar N=(0,a)\times M,\quad\bar h=-dt^2+g_t,
\end{align*}
where $g_t=G^2(t)\,d\Omega^2$ and $(M,d\Omega^2)$ is the unit $3$-sphere. For each $t\in(0,a)$, consider the initial data $(M_t=\{t\}\times M,g_t,K_t)$, where $K_t$, the second fundamental form, is given by 
\begin{align*}
K_t=\frac{\dot{G}(t)}{G(t)}\,g_t.
\end{align*}
In particular, either $K_t$ or $-K_t$ is $2$-convex, depending on the sign of $\dot{G}(t)$. One easily verifies that the DEC holds (strictly) for any choice of scale factor $G(t)$.

For each $t\in(0,a)$, it is easy to see that $(M_t,g_t,K_t)$ contains a spherical MOTS. Indeed, the latitudinal $2$-spheres take on all mean curvature values between $-\infty$ and $+\infty$. Choose the latitudinal $2$-sphere $\Sigma_t$ such that its mean curvature satisfies 
\begin{align}\label{motscondition}
H_t=-\tr_{\Sigma_t}K_t=-2\frac{\dot{G}(t)}{G(t)}.
\end{align}
Then, by \eqref{eq:MOTS.def}, $\Sigma_t$ is a MOTS, $\theta^+_t=0$. 

In fact, it is also the case that $(M_t,g_t,K_t)$ contains a toroidal MOTS. Here, we rely on the one-parameter family of Clifford tori $T_r$ in the unit $3$-sphere $S^3$. By identifying $S^3$ with the unit sphere centered at the origin in $\R^4$, $T_r$, $0<r<1$, is defined as
\begin{align*}
T_r=\left\{(x,y,u,v)\in S^3:x^2+y^2=r^2,\,u^2+v^2=1-r^2\right\}.
\end{align*}
The `standard' Clifford torus is obtained by setting $r = \frac1{\sqrt{2}}$. An elementary computation shows that each $T_r$ has constant mean curvature (see \cite{Kit}),
\begin{align*}
H_r=\frac{1-2r^2}{r\sqrt{1-r^2}}.
\end{align*}
In particular, the Clifford tori take on all mean curvature values between $-\infty$ and~$+\infty$. Thus, arguing as above in the sphere case, there exists an embedded torus $\Sigma_t$ in $(M_t,g_t,K_t)$ satisfying \eqref{motscondition}, which hence is a MOTS. 

One can modify the initial data set $(M_t,g_t,K_t)$ by adding a handle from one side of the torus $\Sigma_t$ to the other, \sl{à la} Gromov-Lawson \cite{GL}, so that $\Sigma_t$ is no longer homologically trivial, and such that the DEC still holds. However, the resulting initial data manifold won't be retractable with respect to $\Sigma_t$, as follows from the next theorem.

\begin{thm}\label{thm:cor.boundary.2}
Let $(M,g,K)$ be an $n$-dimensional, $3\le n\le 7$, closed initial data set satisfying the DEC, $\mu\ge|J|$. Suppose that $(M,g,K)$ admits a MOTS $\Sigma$, with respect to a unit normal field $\nu$, such that the following conditions hold:
\begin{enumerate}
\item[(I)] $M$ is retractable with respect to $\Sigma$ towards $\nu$ and
\item[(II)] $\Sigma$ satisfies the cohomology condition.
\end{enumerate}
Then $\chi^+=0$ on $\Sigma$ and $\Sigma$ is a flat $(n-1)$-torus with respect to the induced metric. Moreover, the following hold:
\begin{enumerate}
\item[(a')] $M\setminus\Sigma\cong(0,\ell)\times\Sigma$ for some $\ell>0$.
\item[] Let $\Sigma_t\cong\{t\}\times\Sigma$ with unit normal $\nu_t$ in direction of the foliation.
\item[(b')] $\chi^+=0$ on $\Sigma_t$ for every $t\in(0,\ell)$.
\item[(c')] $\Sigma_t$ is a flat $(n-1)$-torus with respect to the induced metric for every $t\in(0,\ell)$.
\item[(d')] $\mu+J(\nu_t)=0$ on $\Sigma_t$ for every $t\in(0,\ell)$. In particular, $\mu=|J|$ on $M$.
\end{enumerate}
If we assume further that $K$ is $(n-1)$-convex, we also have:
\begin{enumerate}
\item[(e')] $(M,g)$ is isometric to $[0,\ell]\times\Sigma/{\sim}$ endowed with the induced metric from the product $([0,\ell]\times\Sigma,dt^2+h)$, where $h$ is the induced metric on $\Sigma$. In particular, $(M,g)$ is flat.
\item[(f')] $K=a\,dt^2$, where $a\in C^\infty(M)$ depends only on $t$.
\item[(g')] $\mu=0$ and $J=0$ on $M$.
\end{enumerate}
\end{thm}

\begin{proof}
As in the proof of Theorem~\ref{thm:cor.boundary}, let $(M',g',K')$ be the initial data set derived from $(M,g,K)$ - by making a `cut' along $\Sigma$ - with two boundary components, $\Sigma_0$ and $S$, both isometric to $\Sigma$, such that $\Sigma_0$ is a MOTS with respect to the normal that points into $M'$ and $S$ is a MOTS with respect to the normal that points out of~$M'$. 

It is not difficult to see that $(M',g',K')$ satisfies all the assumptions of Theorem~\ref{thm:EicGalMen} and then all its conclusions. Thus $\Sigma$ is a flat $(n-1)$-torus with $\chi^+=0$ on it and conclusions (a')-(d') of the theorem hold.

If $K$ is $(n-1)$-convex, since $\A(\Sigma_0)=\A(S)$, it follows that $(M',g',K')$ satisfies all the hypotheses of Theorem~\ref{thm:Brane.1} for $\epsilon=0$. Conclusions (e')-(g') then follow.
\end{proof}

\begin{rmk}
It follows, for example, that in a $4$-dimensional spacetime which satisfies the DEC {\it strictly} and which has toroidal Cauchy surfaces, there cannot be any homologically nontrivial toroidal MOTS in any Cauchy surface. This applies, in particular, to the time slices in the toroidal ($k=0$) FLRW spacetimes, that satisfy the Einstein equations with dust (zero-pressure perfect fluid) source.

In view of property (g'), to find initial data sets satisfying the assumptions of Theorem~\ref{thm:cor.boundary.2}, one should perhaps consider vacuum spacetimes. A well-known class of examples are the toroidal Kasner spacetimes,
\begin{align*}
\bar N=(0,\infty)\times M,\quad\bar h=-dt^2+t^{2p_1}dx^2+t^{2p_2}dy^2+t^{2p_3}dz^2,
\end{align*}
where $x,y,z$ are to be understood as periodic coordinates, and where $p_1\le p_2\le p_3$ must satisfy,
\begin{align*}
p_1+p_2+p_3=1\quad\text{and}\quad p_1^2+p_2^2+p_3^2=1.
\end{align*}
Let $M_0$ be the $t=1$ time slice, with metric $g$ and second fundamental form $K$ induced from $(\bar N,\bar h)$. It is not hard to show that in order for $K$ to be $2$-convex, one must have, $p_1=p_2=0$ and $p_3=1$, so that $\bar h$ becomes,
\begin{align*}
\bar h=-dt^2+dx^2+dy^2+t^2dz^2.
\end{align*}
This is an exceptional Kasner spacetime, known as `flat Kasner'. It is locally isometric to Minkowski space. Taking $\Sigma$ to be the torus $t=1$, $z=z_0$, we see that $M_0$ satisfies the assumptions of Theorem~\ref{thm:cor.boundary.2}, including the $2$-convexity assumption.

We mention one further example which illustrates a certain flexibility in initial data sets satisfying (I) and (II), but not the convexity condition. It's a small modification of Example 4.2 in \cite{EicGalMen}. 

Let $\mathbb{R}^3_1$ be Minkowski space with standard coordinates $t,x,y,z$. Consider the box $\mathcal{B}=\{(x,y,z):0\le x\le 1,0\le y\le 1,0\le z\le 1\}$ in the $t=0$ slice. Let $f:\mathcal{B}\to\R$ be a smooth function that vanishes near the boundary of $\mathcal{B}$ and whose graph is spacelike in $\mathbb{R}^3_1$. By identifying opposite sides of the box, we obtain an initial data set $(M,g,K)$ with $M\cong T^3$, where $M$ is given by the graph of $f$, and where $g$ and $K$ are induced from the graph of $f$. Let $\Sigma$ be the intersection of $M$ with the null hyperplane $t=z-\frac12$; see Figure~\ref{dsfig2}. Because the null hyperplane is totally geodesic, $\Sigma$ is necessarily a MOTS. It follows that $(M,g,K)$ satisfies (I) and (II) with respect to $\Sigma$. Note also that $(M,g,K)$ satisfies the DEC; in fact, because it essentially sits in Minkowski space, it is a vacuum initial data set, $\mu=0$, $J=0$. Hence, $(M,g,K)$ satisfies all the assumptions of Theorem~\ref{thm:cor.boundary.2}, except, in general, the convexity condition on $K$. The foliation by MOTS guaranteed by properties (a')-(d') comes from intersecting $M$ with the null hyperplanes $t=z+c$. That these properties hold may be understood in terms of special features of totally geodesic null hypersurfaces. 
\end{rmk}

\begin{figure}[ht]
\begin{center}
\mbox{
\includegraphics[width=4in]{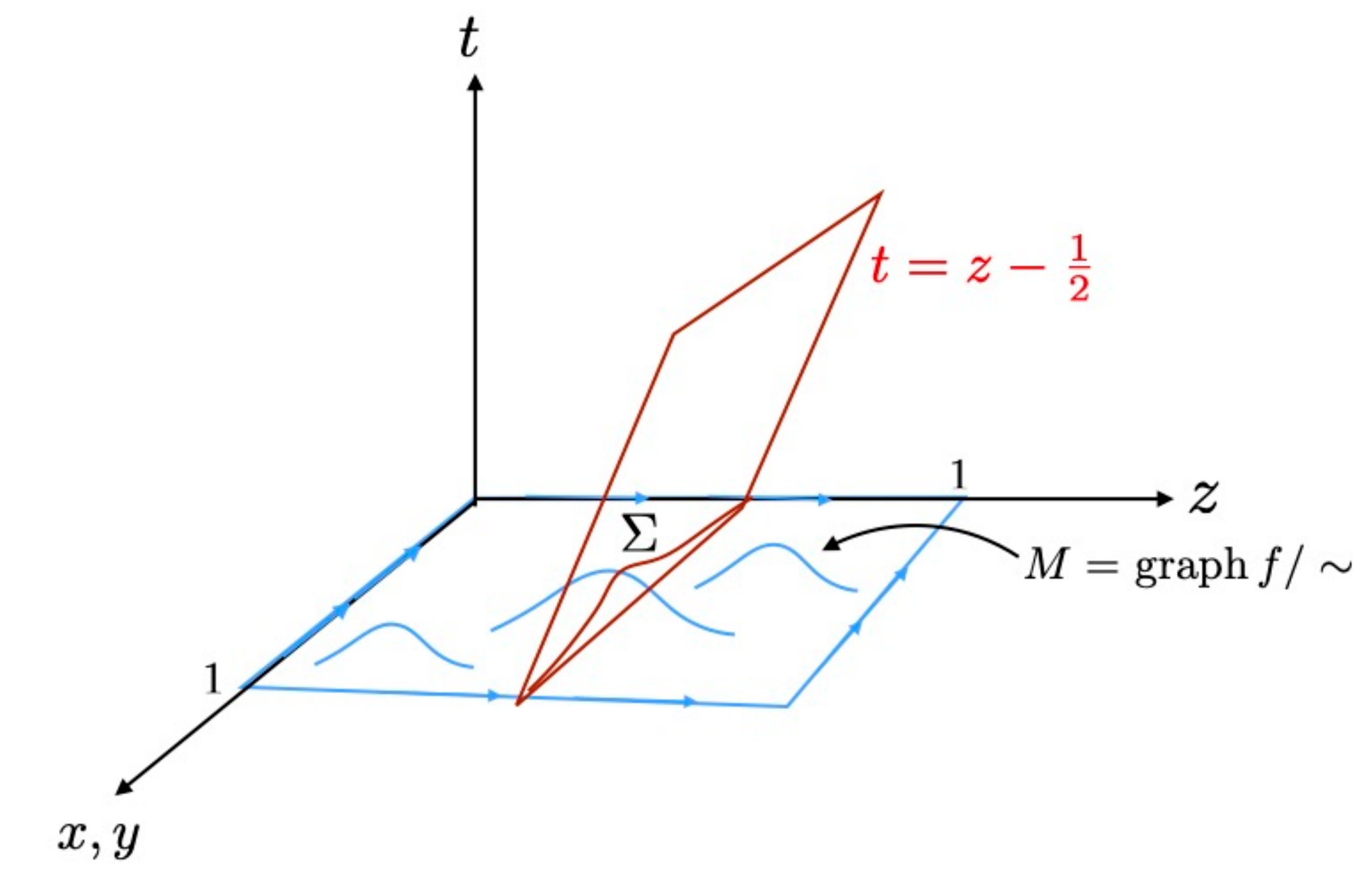}
}
\end{center}
\caption{Initial data set satisfying (I) and (II) of Theorem~\ref{thm:cor.boundary.2}.}
\label{dsfig2}
\end{figure}

\begin{rmk}\label{rem:LohPMT}
Finally, we mention a connection to the spacetime positive mass theorem, specifically the approach taken by Lohkamp~\cite{Lohkamp2016}, from a perspective slightly different from the discussion in \cite{EicGalMen}. Lohkamp reduces the proof to a stand alone result, namely the nonexistence of `$\mu-|J|>0$ islands', see \cite[Theorem~2]{Lohkamp2016}. By a standard compactification (which Lohkamp also considers), the setting of Theorem~2 immediately gives an initial data set satisfying the DEC, with initial data manifold $M\cong T^n\,\sharp\,Q$, $Q$ closed, and a toroidal MOTS $\Sigma$, such that $M$ is retractable with respect to $\Sigma$ (see the discussion at the beginning of Section \ref{sphericalcase}). Theorem~\ref{thm:cor.boundary.2} then yields that $\mu=|J|$ (among other things), which implies Lohkamp's no $\mu-|J|>0$ islands result in dimensions $3\le n\le 7$.
\end{rmk}

Lastly, we consider the following consequence of Theorem~\ref{thm:Brane.2}.

\begin{cor}
Let $(M,g,K)$ be an $n$-dimensional, $3\le n\le 7$, closed initial data set satisfying the DEC, $\mu\ge|J|$. Assume that $-(K+g)$ is $(n-1)$-convex. Then $(M,g,K)$ cannot satisfy conditions {\rm (I)-(II)} of Theorem~\ref{thm:cor.boundary.2}.
\end{cor}

Examples like those discussed at the beginning of Section~\ref{sec:nonsphere} show that, while the conditions (I) and (II) can't be simultaneously satisfied, either one can be.

\begin{proof}
Let $(M',g',K')$ be as in the proof of Theorem~\ref{thm:cor.boundary.2}. If $-(K+g)$ is $(n-1)$-convex, in particular, $-(K'+g')$ is $(n-1)$-convex, it follows from the first part of Theorem~\ref{thm:Brane.2} that $\B(\Sigma_0)\le\B(S)$, which is a contradiction, because
\begin{align*}
\B(\Sigma_0)=\A(\Sigma_0)=\A(S)>\B(S).
\end{align*}
\end{proof}

\bibliographystyle{amsplain}
\bibliography{bibliography}

\end{document}